\documentclass[runningheads]{llncs}
\usepackage{lipsum} 

\usepackage{graphicx}
\usepackage{hyperref}

\usepackage{optidef} 
\usepackage[T1]{fontenc} 
\usepackage[utf8]{inputenc} 
\usepackage{microtype}
\usepackage{placeins}
\usepackage{comment}
\usepackage{enumitem} 
\usepackage{amsmath,amssymb,mathtools,bm} 
\usepackage{varioref} 
\usepackage{comment}
\usepackage{pgfplots,tikz}
\usetikzlibrary{positioning}
\tikzset{main node/.style={circle,fill=blue!20,draw,minimum size=0.5cm,inner sep=0pt},
            }
\usetikzlibrary{arrows}
\pgfplotsset{compat=1.16}
\usetikzlibrary{calc}
\usetikzlibrary{arrows.meta,positioning}
\usepackage{xfrac}
\usepackage{todonotes}
\usepackage{algorithm}
\usepackage{algorithmic}

\DeclarePairedDelimiterX\set[1]\lbrace\rbrace{#1}
\providecommand{\abs}[1]{\lvert#1\rvert}

\DeclareMathOperator*{\argmax}{argmax}

\DeclareMathOperator*{\spt}{spt}
\DeclareMathOperator*{\sgn}{sgn}

\DeclareMathOperator{\R}{\mathbb{R}}

\newcommand{\bx}{\mathbf{x}}
\newcommand{\bX}{\mathbf{X}}
\newcommand{\bM}{\mathbf{M}}

\newcommand{\bp}{\mathbf{p}}
\newcommand{\bP}{\mathbf{P}}

\newcommand{\bU}{\mathbf{U}}

\newcommand{\calG}{\mathcal{G}}

\newcommand{\eps}{\epsilon}

\newcommand{\cl}{\operatorname{cl}}

\def\clap#1{\hbox to 0pt{\hss#1\hss}}

\begin{document}
\title{Multiple Oracle Algorithm to Solve Continuous Games \thanks{The authors acknowledge the support by the project \emph{Research Center for Informatics} (CZ.02.1.01/0.0/0.0/16\_019/0000765).}}

%
%

\author{Tomáš Kroupa$^{\dagger}$\orcidID{0000-0003-1531-2990} \and
Tomáš Votroubek\orcidID{0000-0001-6781-5560}}
\authorrunning{T. Kroupa and T. Votroubek}

\institute{Artificial Intelligence Center, Department of~Computer Science, Faculty of Electrical Engineering, Czech Technical University in Prague, Czech Republic \\
\email{\{tomas.kroupa,votroto1\}@fel.cvut.cz} \\
\url{https://aic.fel.cvut.cz/}}
\maketitle              
\begin{abstract}
  Continuous games are multiplayer games in which strategy sets are compact and utility functions are continuous. These games typically have a highly complicated structure of Nash equilibria, and numerical methods for the equilibrium computation are known only for particular classes of continuous games, such as two-player polynomial games or~games in which pure equilibria are guaranteed to exist. This contribution focuses on the computation and approximation of a mixed strategy equilibrium for the whole class of multiplayer general-sum continuous games. We vastly extend the scope of applicability of the double oracle algorithm, initially designed and proved to converge only for two-player zero-sum games. Specifically, we propose an iterative strategy generation technique, which splits the original problem into the master problem with only a finite subset of strategies being considered, and the subproblem in which an oracle finds the best response of each player. This simple method is guaranteed to recover an approximate equilibrium in finitely many iterations.
  Further, we argue that the Wasserstein distance (the earth mover's distance) is the right metric for the space of mixed strategies for our purposes. Our main result is the convergence of this algorithm in the Wasserstein distance to an equilibrium of the original continuous game. The numerical experiments show the performance of our method on several classes of games including randomly generated examples.

\end{abstract}

\section{Introduction}
A strategic $n$-player game is called continuous if the action space of each player is a compact subset of Euclidean space and all utility functions are continuous. Many application domains have a~continuum of actions expressing the amount of time, resources, location in space~\cite{Kamra18}, or parameters of classifiers \cite{Loiseau19}. This involves also several games modeling the cybersecurity scenaria; see \cite{Niu2021,Loiseau19,roussillon2021scalable}. Continuous games have equilibria in mixed strategies by Glicksberg's generalization of the Nash's theorem \cite{Glicksberg52}, but those equilibria are usually very hard to characterize and compute. We point out the main difficulties in the analysis and development of algorithms and numerical methods for continuous games.
\begin{itemize}
  \item The equilibrium can be any tuple of mixed strategies with infinite supports or almost any tuple of~finitely-supported mixed strategies \cite{Rehbeck18}.
  \item Bounds on the size of supports of equilibrium strategies are known only for particular classes of continuous games \cite{SteinOzdaglarParrilo08}.
\item Some important games have only mixed equilibria; for example, certain variants of Colonel Blotto games \cite{Golman09}.
  \item Finding a mixed strategy equilibrium involves locating its support, which lies inside a continuum of points.
\end{itemize}

To the best of our knowledge, algorithms or numerical methods exist only for very special classes of continuous games. In particular, two-player zero-sum polynomial games can be solved by the sum-of-squares optimization based on the sequence of semidefinite relaxations of the original problem \cite{ParriloIEEE06,LarakiLasserre12}. The book \cite{BasarOlsder99} contains a detailed analysis of equilibria for some families of games with a particular shape of utility functions (games of timing, bell-shaped kernels, etc.) Fictitious play, one of the principal learning methods for finite games, was recently extended to continuous action spaces and applied to Blotto games \cite{Ganzfried21}. The dynamics of fictitious play were analyzed only under further restrictive assumptions in the continuous setting \cite{Hofbauer06}. No-regret learning studied in \cite{Mertikopoulos2019} can be applied to finding pure equilibria or to mixed strategy learning in finite games. Convergence guaranteess for algorithms in the distributed environment solving convex-concave games and some generalizations thereof are developed in \cite{Mertikopoulos2018-ef}. In a similar setting, \cite{Chasnov20} provide convergence guarantees to a neighborhood of a stable Nash equilibrium for gradient-based learning algorithms.

The double oracle algorithm \cite{McMahan03} was extended from finite games and proved to converge for all two-player zero-sum continuous games \cite{Adam21}. The algorithm is relatively straightforward. It is based on the iterative solution of finite subgames and the subsequent extension of the current strategy sets with best response strategies. Despite its simplicity, this method was successfully adapted to large extensive-form games \cite{Bosansky14}, Bayesian games \cite{Li21}, and security domains with complex policy spaces \cite{Xu21}.

In this contribution, we extend the double oracle algorithm beyond two-player zero-sum continuous games. Our main result guarantees that the new method converges in the Wasserstein distance to an equilibrium for any general-sum $n$-player continuous game. Interestingly enough, it turns out that the Wasserstein distance represents a very natural metric on the space of mixed strategies. We demonstrate the computational performance of our method on selected examples of games appearing in the literature and on randomly generated games.

\section{Basic Notions}
The player set is $N=\{1,\dots,n\}$. Each player $i\in N$ selects a pure strategy $x_i$ from a nonempty compact set $X_i\subseteq \R^{d_i}$, where $d_i$ is a positive integer. Put $$\bX=X_1\times\dots\times X_n.$$ A~pure strategy profile is an $n$-tuple $\bx=(x_1,\dots,x_n)\in \bX$. We assume that each utility function $u_i\colon \bX\to \R$ is continuous. A \emph{continuous game} is the tuple $$\calG=\langle N,(X_i)_{i\in N},(u_i)_{i\in N}\rangle.$$  We say that $\calG$ is \emph{finite} if each strategy set $X_i$ is finite.

Consider nonempty compact subsets $Y_i\subseteq X_i$ for $i\in N$. When each $u_i$ is restricted to $Y_1\times\dots\times Y_n$, the continuous game $\langle N,(Y_i)_{i\in N},(u_i)_{i\in N}\rangle $ is called the \emph{subgame} of~$\calG$.

 A~\emph{mixed strategy} of player $i$ is a~Borel probability measure $p$ over $X_i$. The \emph{support} of $p$ is the compact set defined by $$\spt p = \bigcap \{K\subseteq X\mid K \text{ compact},\, p(K)=1\}.$$
 In~the paper, we construct only the mixed strategies with finite supports. The support of Dirac measure $\delta_x$ is the singleton $\spt \delta_x=\{x\}$, where $x\in X_i$. In~general, when the support $\spt p$ of a mixed strategy $p$ is finite, it means that $p(x)>0$ for all $x\in \spt p$ and $\sum_{x\in \spt p} p(x)=1$. For clarity, we emphasize that a finitely-supported mixed strategy $p$ of player $i$ should be interpreted as the function $p\colon X_i\to [0,1]$ vanishing outside $\spt p$, and not as a vector of probabilities with a fixed dimension. This is because only the former viewpoint enables us to consider the distance between \emph{any} pair of pure strategies in $X_i$, which makes it possible to compute a distance between mixed strategies with \emph{arbitrary} supports.

 The set of all mixed strategies of player $i$ is $M_i$. Define $$\bM= M_1\times\dots\times M_n.$$ For a mixed strategy profile $\bp=(p_1,\dots,p_n)\in \bM$, the~expected utility of player $i\in N$ is
$$
  U_i(\bp) =  \smallint_{\bX} u_i\;\mathrm{d}(p_1\times\dots\times p_n),
$$
where $p_1\times\dots\times p_n$ is the product probability measure. This definition  yields a function $U_i\colon \bM\to \R,$ which can be effectively evaluated only in special cases (for example, when each $\spt p_i$ is finite).
For each $i\in N$, let $$\bM_{-i}=\bigtimes_{\substack{k\in N\\ k\neq i}} M_k.$$ A generic profile of mixed strategies from $\bM_{-i}$ is denoted by $$\bp_{-i}=(p_1,\dots,p_{i-1},p_{i+1},\dots,p_n).$$ If player $i$ uses a pure strategy $x\in X_i$ and the rest of the players play $\bp_{-i}\in\bM_{-i}$, we write simply $U_k(x,\bp_{-i})$ in place of $U_k(\delta_x,\bp_{-i})$ to denote the expected utility  of player~$k$.

A mixed strategy profile $\bp^*=(p_1^*,\dots,p_n^*)\in\bM$ is a \emph{(Nash) equilibrium} in~a~continuous game $\calG$ if, for every $i\in N$ and all $p_i\in M_i$, the inequality $U_i(p_i,\bp_{-i}^*) \leq U_i(\bp^*)$ holds. Glicksberg's theorem \cite{Glicksberg52} says that any continuous game $\calG$ has an equilibrium. The following useful characterization is a consequence of  Proposition \ref{prop:switch}: A profile $\bp^*$ is an equilibrium if, and only if, $U_i(x_i,\bp_{-i}^*) \leq U_i(\bp^*)$, for each $i\in N$ and every $x_i\in X_i$.

Let $\eps\geq 0$. An \textit{$\epsilon$-equilibrium} is a mixed strategy profile $\bp^*\in \bM$ such that $U_i(x_i,\bp_{-i}^*)- U_i(\bp^*) \leq \eps$, for all $i\in N$ and $x_i\in X_i$. This implies that, for every $p_i\in M_i$, the inequality $U_i(p_i,\bp_{-i}^*)- U_i(\bp^*)\le \eps$ holds, too. Let $$\bU(\bp)=(U_1(\bp),\dots,U_n(\bp)).$$ Define $$\bU(\bx,\bp)=(U_1(x_1,\bp_{-1}),\dots,U_n(x_n,\bp_{-n})),$$ where $\bx\in\bX$ and $\bp\in \bM$. Let $\bm{\eps}=(\eps,\dots,\eps)$. Using the above introduced vectorial notation, a mixed strategy profile $\bp^*\in \bM$ is an $\epsilon$-equilibrium if, and only if, $$\bU(\bx,\bp^*)-\bU(\bp^*)\leq \bm{\eps},\qquad \text{for all $\bx\in\bX$}.$$

\section{Convergence of Mixed Strategies}
 We consider an arbitrary metric~$\rho_i$ on the compact strategy space $X_i\subseteq \R^{m_i}$ of each player $i\in N$. This enables us to quantify a~distance between pure strategies $x,y \in X_i$ by the number $\rho_i(x,y)\ge 0$. Consequently, we can define the Wasserstein distance $d_W$ on $M_i$, which is compatible with the metric of the underlying strategy space $X_i$ in the sense that $$d_W(\delta_x,\delta_y)=\rho_i(x,y),\qquad \text{for all $x,y \in X_i$}.$$ The preservation of distance from $X_i$ to $M_i$ is a very natural property, since the space of pure strategies $X_i$ is embedded in $M_i$ via the correspondence $x\mapsto \delta_x$ mapping the pure strategy $x$ to the Dirac measure $\delta_x$.

The Wasserstein distance originated from optimal transport theory. It is nowadays highly instrumental in solving many problems of computer science. Specifically, the \emph{Wasserstein distance} of mixed strategies $p,q\in M_i$ is $$
d_W(p,q) = \inf_{\mu} \int_{X_i^2} \rho_i(x,y) \;\mathrm{d}\mu(x,y),$$
where the infimum is over all Borel probability measures $\mu$ on $X_i^2$ whose one-dimensional marginals are $p$ and $q$:
\begin{align*}
  \mu(A\times X_i) & = p(A),\\
  \mu(X_i \times A)  & = q(A), \qquad \text{for all Borel subsets $A\subseteq X_i$.}
\end{align*}

The dependence of $d_W$ on the metric $\rho_i$ is understood. The function $d_W$ is a metric on $M_i$. Since $X_i$ is compact, it has necessarily bounded diameter. This implies that the convergence in $(M_i,d_W)$ coincides with the weak convergence \cite[Corollary 2.2.2]{Panaretos20}. Specifically, the following two assertions are equivalent for any sequence $(p^j)$ in $M_i$:
\begin{enumerate}
  \item $(p^j)$ converges to $p$ in $(M_i,d_W)$.
  \item $(p^j)$ \emph{weakly converges} to $p$, which means by the definition that $$\lim_{j} \smallint\nolimits_{X_i}f\;\mathrm{d}p^j =  \smallint\nolimits_{X_i}f\;\mathrm{d}p,\qquad \text{for every continuous function $f\colon X_i\to \R$.}$$
\end{enumerate}

The metric space $(M_i,d_W)$ is compact by \cite[Proposition 2.2.3]{Panaretos20}. Consequently, the joint strategy space $\bM$ is compact in a product metric as well, and any sequence $(\bp^j)$ in $\bM$ has an accumulation point. In other words:
\begin{proposition}\label{prop:convergence}
  Any sequence of mixed strategy profiles $(\bp^j)$ in $\bM$ contains a~weakly convergent subsequence.
\end{proposition}
\noindent
 The function $U_i$ is continuous on $\bM$ by the definition of weak convergence. This implies that if $(\bp^j)$ weakly converges to $\bp$ in $\bM$, then the corresponding values of expected utility goes to $U_i(\bp)$, that is, $\lim_j U_i(\bp^j) = U_i(\bp)$.
By compactness of~$M_i$ and continuity of $U_i$, all maxima and maximizers in the paper exist. In particular, the optimal value of utility function in response to the mixed strategies of other players is attained for some pure strategy.

\begin{proposition}\label{prop:switch}
For each player $i\in N$ and any mixed strategy profile $\bp_{-i}\in \bM_{-i}$, there exists a pure strategy $x_i\in X_i$ such that $$\max_{p\in M_i} U_i(p,\bp_{-i})=\max_{x\in X_i} U_i(x,\bp_{-i})=U_i(x_i,\bp_{i}).$$
\end{proposition}

In general, computing $d_W(p,q)$ for any pair of mixed strategies $p,q\in M_i$ is a difficult infinite-dimensional optimization problem. The existing numerical methods for this problem are reviewed in \cite{Peyre19}. In our setting it suffices to evaluate $d_W(p,q)$ only for mixed strategies with finite supports. In particular, it follows immediatelly from the definition of $d_W$ that $$d_W(p,\delta_y)=\sum_{x\in \spt p}\rho_i(x,y)p(x),$$ for every finitely-supported mixed strategy $p\in M_i$ and any $y \in X_i$. If mixed strategies $p$ and $q$ have finite supports, then computing~$d_W(p,q)$ becomes the linear programming problem with variables $\mu(x,y)$ indexed by $(x,y)\in\spt p \times \spt q$. Specifically, the objective function to be minimized is
\begin{equation} \label{def:LPobj}
  \sum_{x\in \spt p} \sum_{y\in \spt q} \rho_i(x,y) \mu(x,y),
\end{equation}
and the constraints are
\begin{align*}
  \mu(x,y) & \ge 0, \quad \forall (x,y)\in\spt p \times \spt q,\\
  \sum_{x\in \spt p} \sum_{y\in \spt q}\mu(x,y) & =1, \\
  \sum_{y\in \spt q}\mu(x,y) & = p(x),\quad \forall x\in \spt p \\
  \sum_{x\in \spt p}\mu(x,y) & = q(y),\quad \forall y\in \spt q.
\end{align*}

We will briefly mention one of the frequently used alternatives to the Wasserstein distance. The \emph{total variation distance} between mixed strategies $p,q\in M_i$ of player $i$ is $$d_{TV}(p,q)= \sup_{A}\; \abs{p(A)-q(A)},$$ where the supremum is over all Borel subsets $A\subseteq X_i$. When both mixed strategies $p$ and $q$ have finite supports, we have $d_{TV}(p,q)=\frac 12\sum_{x}\abs{p(x)-q(x)}$, where the sum is over $\spt p\;\cup\; \spt q$. In this case, $d_W$ and $d_{TV}$ satisfy the following inequalities (see \cite{Gibbs02}):
\begin{equation}\label{W:ineq}
  d_{min} \cdot  d_{TV}(p,q)\le d_W(p,q) \le d_{max} \cdot  d_{TV}(p,q),
\end{equation}
where
\begin{align*}
  d_{min} & = \min \{\rho_i(x,y) \mid x,y \in \spt p \cup \spt q, x\ne y\},\\
  d_{max} & = \max \{\rho_i(x,y) \mid x,y \in \spt p \cup \spt q\}.
\end{align*}

If $(p^j)$ converges to $p$ in the total variation distance, then $(p^j)$ converges weakly, but the converse fails. For example, if a sequence $(x^j)$ converges to $x\in X_i$ and $x^j\neq x$ for all~$j$, then the corresponding sequence of Dirac measures converges weakly in $M_i$, as $d_W(\delta_{x^j},\delta_x)=\rho_i(x^j,x)\to 0$. By contrast, it fails to converge in the total variation distance, since $d_{TV}(\delta_{x^j},\delta_{x})=1$ for all $j$.

\section{Main Results}

\begin{algorithm}[t]
  \caption{}
  \label{alg:do}
\textbf{Input}: Continuous game $\calG=\langle N,(X_i)_{i\in N},(u_i)_{i\in N}\rangle$, nonempty finite subsets of initial strategies $X_1^1\subseteq X_1,\dots,X_n^1\subseteq X_n$, and $\epsilon\geq 0$\\
\textbf{Output}: $\eps$-equilibrium $\bp^j$ of game $\calG$.

 \begin{algorithmic}[1]
    \STATE $j= 0$
    \REPEAT
    \STATE $j= j+1$
    \STATE Find an equilibrium $\bp^j$ of $\langle N,(X^j_i)_{i\in N},(u_i)_{i\in N}\rangle$
    \FOR{$i\in N$}
    \STATE Find some $x^{j+1}_i\in \beta_i(\bp_{-i}^j)$
    \STATE $X^{j+1}_i= X_i^j \cup \{x^{j+1}_i\}$
    \ENDFOR
    \UNTIL{$\bU(\bx^{j+1},\bp^j)-\bU(\bp^j)\leq \bm{\eps}$}
 \end{algorithmic}
\end{algorithm}

We propose Algorithm \ref{alg:do} as an iterative strategy-generation technique for (approximately) solving any continuous game~$\calG=\langle N,(X_i)_{i\in N},(u_i)_{i\in N}\rangle$. We recall that the~\emph{best response set} of player $i\in N$ with respect to a mixed strategy profile $\bp_{-i}\in \bP_{-i}$ is
$$
\beta_i(\bp_{-i}) =  \argmax_{x\in X_i} U_i(x,\bp_{-i}).
$$
The set $\beta_i(\bp_{-i})$ is always nonempty by Proposition \ref{prop:switch}.  We assume that every player uses an oracle to recover at least one best response strategy, which means that the player is able to solve the corresponding optimization problem to global optimality. In Section \ref{sec:exp} we will see the instances of games for which this is possible.

Algorithm \ref{alg:do} proceeds as follows. In every iteration $j$, finite strategy sets $X_i^j$ are constructed for each player $i\in N$ and the corresponding finite subgame of $\calG$ is solved. Let $\bp^j$ be its equilibrium. Then, an arbitrary best response strategy $x_i^{j+1}$ with respect to $\bp_{-i}^{j}$ is added to each strategy set~$X_i^j$. Those steps are repeated until
\begin{equation}\label{def:diff}
  U_i(x_i^{j+1},\bp_{-i}^j)-U_i(\bp^j)\leq \eps \qquad \text{for each $i\in N$.}
\end{equation}

First we discuss basic properties of the algorithm.  At each step~$j$, we have the inequality
\begin{equation}\label{ineq:upper}
  \bU(\bx^{j+1},\bp^j)\ge \bU(\bp^j).
\end{equation}
Indeed, for each player $i\in N$, we get
$$
U_i(x_i^{j+1},\bp_{-i}^j) = \max_{x \in X_i} U_i(x,\bp^j_{-i})
 \ge \max_{x \in X_i^j} U_i(x,\bp^j_{-i}) = U_i(\bp^j).$$

We note that the stopping condition of the double oracle algorithm for two-player zero-sum continuous games \cite{Adam21} is necessarily different from \eqref{def:diff}. Namely the~former condition, which is tailored to the zero-sum games, is
\begin{equation} \label{stoppingDO}
  U_1(x_1^{j+1},p_2^j) - U_1(x_2^{j+1},p_1^j) \le \eps.
\end{equation}
When $\calG$ is a two-player zero-sum continuous game, it follows immediately from \eqref{ineq:upper} that \eqref{stoppingDO} implies \eqref{def:diff}.

We prove correctness of Algorithm \ref{lem:algstop} --- the eventual output $\bp^j$ is an $\eps$-equilibrium of the original game $\calG$.
\begin{lemma}\label{lem:algstop}
The strategy profile $\bp^j$ is an $\eps$-equilibrium of $\calG$, whenever Algorithm~\ref{alg:do} terminates at step $j$.
\end{lemma}
\begin{proof}
 Let $\bx\in \bX$. We get
 $$
 \bU(\bx,\bp^j)-\bU(\bp^j) = \bU(\bx,\bp^j)-\bU(\bx^{j+1},\bp^j) +  \bU(\bx^{j+1},\bp^j)-\bU(\bp^j).
 $$
Then $\bU(\bx,\bp^j)-\bU(\bx^{j+1},\bp^j)\leq \mathbf{0}$, since $\bx^{j+1}$ is the profile of best response strategies. Consequently, we obtain
$$
\bU(\bx,\bp^j)-\bU(\bp^j)  \le  \bU(\bx^{j+1},\bp^j)-\bU(\bp^j) \le \bm{\eps},
$$
where the last inequality is just the terminating condition \eqref{def:diff}. Therefore, $\bp^j$ is an $\eps$-equilibrium of  $\calG$. \qed
\end{proof}

If the set $\bX^j=X^j_1\times \dots \times X^j_n$ cannot be further inflated, Algorithm \ref{alg:do} terminates.
\begin{lemma}\label{lemma:do_stop}
  Assume that $\bX^j=\bX^{j+1}$ at step $j$ of Algorithm~\ref{alg:do}. Then $$ \bU(\bx^{j+1},\bp^j)=\bU(\bp^j).$$
\end{lemma}
\begin{proof}
  The condition $\bX^j=\bX^{j+1}$ implies $\bx^{j+1}\in \bX^{j}$. Then
  $$
U_i(x_i^{j+1},\bp^j_{-i})=\max\limits_{x\in X_i^j}U_i(x,\bp^j_{-i}) =U_i(\bp^j),
$$
for all $i\in N$. \qed
\end{proof}
Proposition \ref{pro:Mahan} extends the analogous result for two-player zero-sum finite games from \cite{McMahan03}.
\begin{proposition}\label{pro:Mahan}
  If $\calG$ is a finite game and $\eps=0$, then Algorithm \ref{alg:do} recovers an equilibrium of $\calG$ in finitely-many steps.
\end{proposition}
\begin{proof}
  Let $\calG$ be finite and $\eps=0$. Since each $X_i$ is finite, there exists an iteration $j$  in which $\bX^{j+1}=\bX^j$. Then the terminating condition of Algorithm \ref{alg:do} is satisfied with $\eps=0$ (Lemma \ref{lemma:do_stop}) and $\bp^j$ is an equilibrium of $\calG$ (Lemma~\ref{lem:algstop}). \qed
\end{proof}
\noindent
This is our main result, which is a non-trivial extension of the convergence theorem from \cite{Adam21}.
\begin{theorem}\label{thm:convergence}
Let $\calG$ be a continuous game.
\begin{enumerate}
  \item Let $\eps=0$. If Algorithm \ref{alg:do} stops at step $j$, then $\bp^j$ is an~equilibrium of $\calG$. Otherwise, any accumulation point of $\bp^1,\bp^2,\dots$ is an equilibrium of $\calG$.
 \item Let $\eps>0$. Then Algorithm \ref{alg:do} terminates at some step $j$ and $\bp^j$ is an $\eps$-equilibrium of~$\calG$.
\end{enumerate}
\end{theorem}
\begin{proof}

  \emph{Item 1.} Let $\eps=0$. If Algorithm~\ref{alg:do} terminates at step $j$, then Lemma \ref{lem:algstop} implies that $\bp^j$ is an equilibrium of $\calG$. In the opposite case, the algorithm generates a sequence of mixed strategy profiles $\bp^1, \bp^2, \dots$ Consider any weakly convergent subsequence of this sequence --- at least one such subsequence exists by Proposition~\ref{prop:convergence}. Without loss of generality, such a subsequence will be denoted by the same indices as the original sequence. Therefore, for each player $i\in N$, there exists some $p^*_i\in M_i$, such that the sequence $p_i^1,p_i^2,\dots$ weakly converges to $p_i^*$. We need to show that  $\bp^*=(p^*_1,\dots,p^*_n)\in \bM$ is an equilibrium of $\calG$.

  Let $i\in N$. Define
  $$
  Y_i = \bigcup_{j=1}^{\infty} X^j_i.
  $$
     First, assume that $x\in Y_i$. Then there exists $j_0$ with $x\in X_i^j$ for each $j\ge j_0$. Hence the inequality $U_i(\bp^j) \geq U_i(x,\bp_{-i}^j)$, for each $j\ge j_0$, since $\bp^j$ is an equilibrium of the corresponding finite subgame. Therefore,
  $$
  U_i(\bp^*) = \lim_{j}U_i(\bp^j) \geq \lim_{j} U_i(x,\bp_{-i}^j) = U_i(x,\bp_{-i}^*).
  $$
  Further, by continuity of $U_i$,
  \begin{equation}\label{ineq:NExcl}
    U_i(\bp^*) \geq  U_i(x,\bp_{-i}^*) \quad \text{for all } x\in\cl(Y_i).
  \end{equation}
 Now, consider an arbitrary $x\in X_i\setminus \cl(Y_i)$. The definition of $x_{i}^{j+1}$ yields $U_i(x_{i}^{j+1},\bp^j_{-i}) \geq U_i(x,\bp^j_{-i})$ for each $j$. This implies, by continuity,
  \begin{equation}\label{eq:conv_aux3}
    \lim_{j} U_i(x_{i}^{j+1},\bp_{-i}^*) \geq \lim_{j} U_i(x,\bp^j_{-i}) =  U_i(x,\bp_{-i}^*).
  \end{equation}
  Since $x_{i}^{j+1}\in X_i^{j+1}$, compactness of $X_i$ provides a convergent subsequence (denoted by the same indices) such that $ x'= \lim\limits_j x_i^{j} \in \cl(Y_i)$. Then (\ref{ineq:NExcl}) gives
  \begin{equation}\label{ineq:NEx}
    U_i(\bp^*)\ge U_i(x',\bp_{-i}^*) =\lim_{j} U_i(x_{i}^{j+1},\bp_{-i}^*).
  \end{equation}
Combining (\ref{eq:conv_aux3}) and (\ref{ineq:NEx}) shows that $U_i(\bp^*) \geq  U_i(x,\bp_{-i}^*)$.

\emph{Item 2.} Let $\eps>0$. If Algorithm~\ref{alg:do} terminates at step~$j$, then Lemma \ref{lem:algstop} implies that $\bp^j$ is an $\eps$-equilibrium of $\calG$. Otherwise Algorithm ~\ref{alg:do} produces a sequence $\bp^1,\bp^2,\dots$ and we can repeat the analysis as in Item 1 for  convergent subsequences of $(\bp^j)$ and $(\bx^j)$, which are denoted by the same indices. Define $x'= \lim_j x_i^j.$ Then
\begin{equation}\label{ineq:aux1}
  U_i(\bp^*)\ge U_i(x',\bp_{-i}^*)= \lim_{j} U_i(x_{i}^{j+1},\bp_{-i}^j).
\end{equation}
At every step $j$ we have
$U_i(x_i^{j+1},\bp^j_{-i})\geq U_i(\bp^j)$ for each $i\in N$ by \eqref{ineq:upper}. Hence
$
  \lim_{j} U_i(x_{i}^{j+1},\bp_{-i}^j) \geq  U_i(\bp^*).
$
Putting together the last inequality with (\ref{ineq:aux1}), we get
\begin{equation}\label{eq:limiting}
  \lim_{j} (U_i(x_{i}^{j+1},\bp_{-i}^j)-U_i(\bp^j)) = 0
\end{equation}
for each $i\in N$. This equality means that Algorithm \ref{alg:do} stops at some step $j$ and $\bp^j$ is an $\eps$-equilibrium by Lemma \ref{lem:algstop}. \qed
\end{proof}

Algorithm \ref{alg:do} generates the sequence of equilibria $\bp^1,\bp^2,\dots$ in increasingly larger subgames of $\calG$. The sequence itself may fail to converge weakly in $\bM$ even for a~two-player zero-sum continuous game; see Example 1 from \cite{Adam21}. In fact, Theorem \ref{thm:convergence} guarantees only convergence to an accumulation point. We recall that this is a~typical feature of some globally convergent methods not only in infinite-dimensional spaces \cite[Theorem 2.2]{Hinze08}, but also in Euclidean spaces. For example, a gradient method generates the sequence such that only its accumulation points are guaranteed to be stationary points; see \cite[Proposition 1.2.1]{Bertsekas16} for details. One necessary condition for the weak convergence of $\bp^1,\bp^2,\dots$ is easy to formulate using the stopping condition of Algorithm~\ref{alg:do}.
\begin{proposition}
  If the sequence $\bp^1,\bp^2,\dots$ generated by Algorithm \ref{alg:do} converges weakly to an equilibrium $\bp$, then $
  \lim_j \bigl(U_i(x_i^{j+1},\bp_{-i}^j)-U_i(\bp^j)\bigr)=0,$ for each $i\in N.
  $
\end{proposition}
\begin{proof}
   Using \eqref{ineq:upper} and the triangle inequality,
  \begin{align*}
    &U_i(x_i^{j+1},\bp_{-i}^j)-U_i(\bp^j)=\abs{U_i(x_i^{j+1},\bp_{-i}^j)-U_i(\bp^j)} \\
    &\leq \abs{U_i(x_i^{j+1},\bp_{-i}^j)-U_i(\bp)} + \abs{U_i(\bp)-U_i(\bp^j)}.
  \end{align*}
  As $j\to \infty$, the first summand goes to zero by \eqref{eq:limiting} and the second by the assumption. Hence the conclusion. \qed
\end{proof}

In our numerical experiments (see Section \ref{sec:exp}), we compute the difference
\begin{equation}\label{crit1}
  U_i(x_i^{j+1},\bp_{-i}^j)-U_i(\bp^j)
\end{equation}
 at each step $j$ and check if such differences are diminishing with $j$ increasing. This provides a simple heuristics to detect the quality of approximation and convergence. Another option is to calculate the Wasserstein distance
 \begin{equation}\label{crit2}
  d_W(\bp^j,\bp^{j+1}),
 \end{equation}
 which can be done with a linear program \eqref{def:LPobj} or approximately using the bounds \eqref{W:ineq}. If the sequence $\bp^1,\bp^2,\dots$ converges weakly, then $d_W(\bp^j,\bp^{j+1}) \to 0$. We include the values \eqref{crit2} in the results of some numerical experiments and observe that they are decreasing to zero quickly. It can be shown that neither \eqref{crit1} nor \eqref{crit2} are monotone sequences. The lack of monotonicity is apparent from the graphs of our experiments; see Figure \ref{figure0}, for example.

Algorithm \ref{alg:do} is a meta-algorithm, which is parameterized by
\begin{enumerate}
  \item the algorithm for computing equilibria of sampled finite games (the \emph{master problem}) and
  \item the optimization method for computing the best response (the \emph{sub-problem}).
\end{enumerate}
 We detail this setup for each example in the next section. The choice of computational methods should reflect the properties of a~continuous game, since the efficiency of methods for solving the master problem and subproblem is the decisive factor for the overall performance and precision of Algorithm \ref{alg:do}. For example, polymatrix games are solvable in polynomial time \cite{Cai16}, whereas finding even an approximate Nash equilibrium of a finite general-sum game is a very hard problem \cite{Daskalakis09}. As for the solution of the subproblem, the best response computation can be based on global solvers for special classes of utility functions.

\section{Numerical Experiments}\label{sec:exp}
We demonstrate the versatility of our method by solving (i) various games appearing in current papers and (ii) randomly generated games. In some cases we show the progress of the convergence-criterion value \eqref{crit1} called ``instability'' over the course of $10$ iterations, and we also plot the Wasserstein distance \eqref{crit2} between mixed strategies in consecutive iterations. All experiments were initialized with random unit vectors. In the games with polynomial utility functions, the best response oracles employ methods of global polynomial optimization \cite{Lasserre2015}. In other cases we use local solvers, which nevertheless perform sufficiently well.

We used a laptop running Linux 5.13 on Intel Core i5-7200U CPU with 8 GiB of system memory to perform our experiments. Our implementation uses Julia 1.6, JuMP \cite{DunningHuchetteLubin2017}, and Mosek. We also used the solver Ipopt \cite{ipopt} when an explicit best response formulation was unavailable. The Julia source codes will be attached to this paper. Examples \ref{ex0}--\ref{ex2} took between 0.1 and 0.2 seconds to compute and Example \ref{ex3} took $1$ second.

\begin{example}[Zero-sum polynomial game \cite{ParriloIEEE06}]\label{ex0}
    Consider a two-player zero-sum game with strategy sets $[-1,1]$ and with the utility function of the first player $u(x,y) = 2xy^2 - x^2 - y$ on $[-1,1]^2$. As the generated subgames are zero-sum, we can use linear programming to find their equilibria. The global method for optimizing polynomials described in \cite{ParriloIEEE06} is an appropriate best response oracle in this case. After 10 iterations, our method finds pure strategies $
    x \approx 0.4$ and $y \approx 0.63$,
    resulting in payoffs \((-0.47, 0.47)\). An oracle based on a hierarchy of semidefinite relaxations (as implemented in SumOfSquares \cite{weisser2019polynomial}) can be used instead to handle games with semialgebraic strategy sets.

    \begin{figure}[!htpb]
      \centering
        \includegraphics{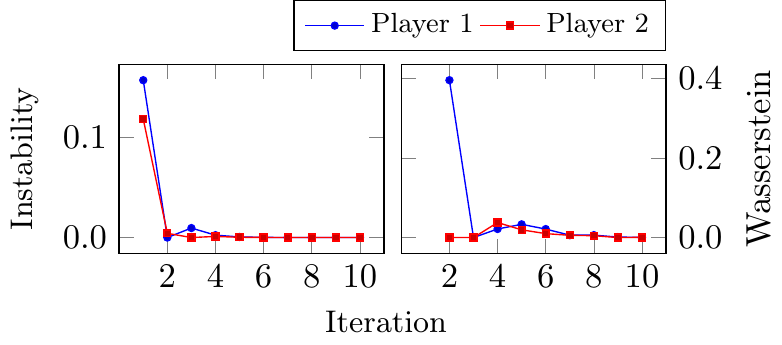}
        \caption{Convergence in Example \ref{ex0}}
        \label{figure0}
    \end{figure}

\end{example}
\begin{example}[General-sum polynomial game \cite{SteinOzdaglarParrilo08}]\label{ex1}
The strategy set of each player is $[-1,1]$ and utility functions are
\begin{align*}
  u_{1}(x,y) & = -3x^2y^2 - 2x^3 + 3y^3 + 2xy - x,\\
  u_{2}(x,y) &= 2x^2y^2 + x^2y - 4y^3 - x^2 + 4y.
\end{align*}
 Our method finds mixed strategies $$
        x \approx\begin{cases}
            0.11 &44.19,\%\\
            -1.0 &55.81,\%\\
        \end{cases} \qquad
        y \approx 0.72,
    $$
    resulting in payoffs \((1.13, 1.81)\). We use the PATH solver \cite{Ferris00} for linear complementarity problems to find equilibria in the generated subgames.

    \begin{figure}[!htpb]
      \centering
        \includegraphics{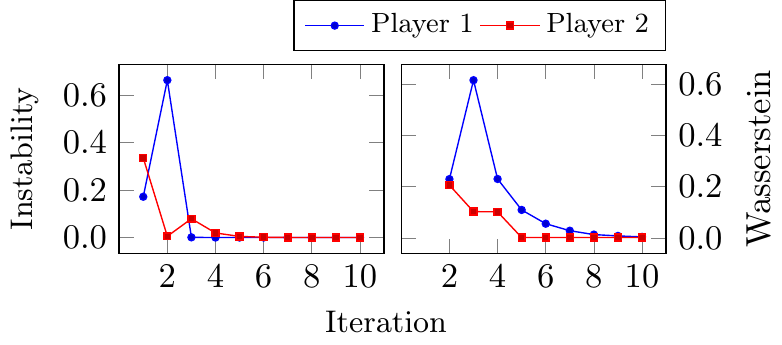}
        \caption{Convergence in Example \ref{ex1}}
    \end{figure}
\end{example}
\begin{example}[Torus game \cite{Chasnov20}]\label{ex2}
    Each strategy set is the unit circle $S^1 = [-\pi, \pi]$ and the utility functions are
 \begin{align*}
  u_{1}(\theta_1, \theta_2) &= \alpha_1 \cos(\theta_1 - \phi_1) - \cos(\theta_1 - \theta_2),\\
  u_{2}(\theta_1, \theta_2) &= \alpha_2 \cos(\theta_2 - \phi_2) - \cos(\theta_2 - \theta_1).
 \end{align*}
 where $\phi = (0, \pi/8)$ and $\alpha = (1, 1.5)$. Using Ipopt as the best response oracle, our method returns pure strategies $\theta_1 \approx 1.41$, $\theta_2 \approx -0.33
    $
    resulting in payoffs \((0.32, 1.29)\).
    \begin{figure}[!htpb]
      \centering
        \includegraphics{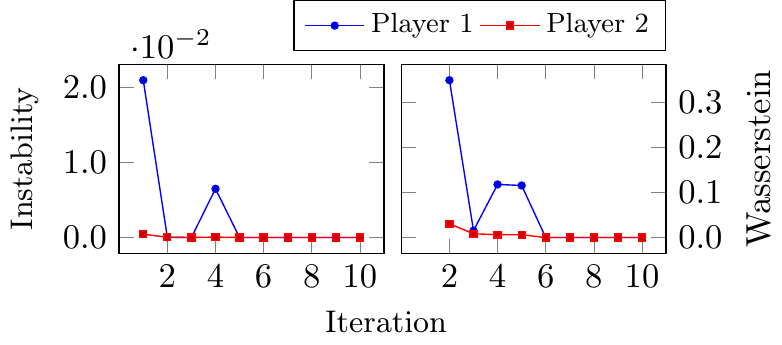}
        \caption{Convergence in Example \ref{ex2}}
    \end{figure}
\end{example}
\begin{example}[General Blotto \cite{Golman09}]\label{ex3}
    Each strategy set in this two-player zero-sum game is the standard $4$-dimensional simplex in $\R^5$ and the utility function of the first player is
    $
    u(\textbf{x}, \textbf{y}) = \sum_{j=1}^5 f(x_j - y_j)$,
      where $f(x) = \sgn(x)\cdot x^2$. Using linear programming to solve the master problem and Ipopt for the approximation of best response, our method finds the pure strategies $
        \textbf{x} \approx (0, 0, 0, 1, 0)$ and
        $\textbf{y} \approx (0, 1, 0, 0, 0)$
    resulting in payoffs \((0, 0)\). This is an~equilibrium by \cite[Proposition 4]{Golman09}.

    \begin{figure}[!htpb]
      \centering
        \includegraphics{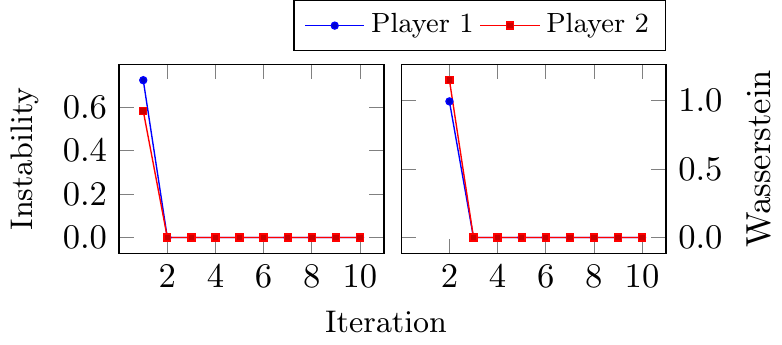}
        \caption{Convergence in Example \ref{ex3}}
    \end{figure}
\end{example}

We note that some well-known classes of games with compact actions spaces cannot be used in our experiments. The typical case in point are Colonel Blotto games \cite{Behnezhad2019-fc} since their utility function is discontinuous. By contrast, certain General Blotto games \cite{Golman09} are continuous games (see Example \ref{ex3}).

\subsection{Experiments with separable network games}
Separable network games (polymatrix games) with finitely many strategies of each player can be solved in polynomial time \cite{Cai16} by linear programming. However, the behaviour of learning methods such Multiplicative Weights Update can be fairly complex already in case of two players; see \cite{Bailey2018-to}. We remark that zero-sum polymatrix games are payoff-equivalent to the general polymatrix games by \cite[Theorem 7]{Cai16}.

We use Algorithm \ref{alg:do} to compute equilibria of polymatrix games defined by 20 by 20 matrices. Specifically, we generated a random matrix for each edge in the network and then transposed and subtracted the matrix of utility functions to make the game globally zero sum. In a test of 100 games, our algorithm found an $\eps$-equilibrium with $\eps = 0.01$ after $17.69$ iterations in a mean time of $0.29$ seconds.

Further, we considered a continuous generalization of separable network games in which the strategy sets are $[-1,1]$ and the utility functions are polynomials. This class of games was analyzed with the tools of polynomial optimization in \cite{KroVaVo}.

\begin{example}[Three-player zero-sum polynomial game]\label{ex4}
 There are 3 players. All pairs of players are involved in bilateral general-sum games and each player uses the same strategy across all such games. The sum of all utility functions is zero. The pairwise utility functions on $[-1,1]^2$ are
     \[\begin{aligned}
         u_{1, 2}(x_1,x_2) &= -2x_1x_2^2 + 5x_1x_2 - x_2\\
         u_{1, 3}(x_1,x_3) &= -2x_1^2 - 4x_1x_3 - 2x_3\\
         u_{2, 1}(x_2,x_1) &= 2x_1x_2^2 - 2x_1^2 - 5x_1x_2 + x_2\\
         u_{2, 3}(x_2,x_3) &= -2x_2x_3^2 - 2x_2^2 + 5x_2x_3\\
         u_{3, 1}(x_3,x_1) &= 4x_1^2 + 4x_1x_3 + 2x_3\\
         u_{3, 2}(x_3,x_2) &= 2x_2x_3^2 + 2x_2^2 - 5x_2x_3\\
     \end{aligned}\]
     The polynomial (sum-of-squares) optimization serves as the best response oracle. Our method finds the mixed strategies
      $$
         x_1 \approx -0.06,\;
         x_2 \approx\begin{cases}
             0.35 &70.8\%\\
             0.36 &29.2\%\\
         \end{cases},\;
         x_3 \approx\begin{cases}
             1.0 &72.11\%\\
             -1.0 &27.89\%\\
         \end{cases}
     $$
     with the corresponding payoffs \((-1.23, 0.26, 0.97)\).
     \begin{figure}[!htpb]
       \centering
         \includegraphics{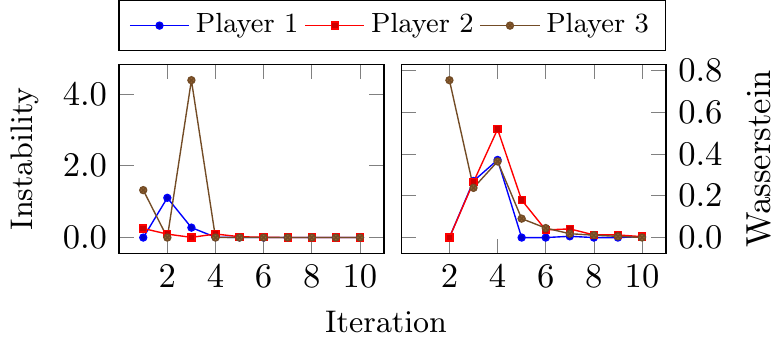}
         \caption{Convergence in Example \ref{ex4}}
     \end{figure}
 \end{example}

In the final round of experiments, we solved randomly generated network games of five players whose utility functions are quartic polynomials over the product of $[-1,1]$ to show that this approach can also solve more complex games. In~particular, we generated the network games by adding three random monomials of degree four or less to each pairwise game, then subtracted the transpose to satisfy the global zero-sum property. While the sum-of-squares approach scales poorly as the polynomial degree and the number of variables grows, its use as a best response oracle means that we only have to consider the variables of one player at~a time. In a test of 100 games, our algorithm found an $\eps$-equilibrium with $\eps = 0.01$ after $5.06$ iterations in $0.78$ seconds on average.

\subsection{Experiments with random general-sum polynomial games}\label{sec:exppolyquartic}
In this experiment, we used the multiple oracle algorithm to find $\eps$-equilibria ($\eps = 0.001$) in continuous multiplayer games with randomly generated quartic polynomial payoffs and $[0,1]$ strategy sets. The time to find equilibria does not appear strongly correlated with the number of players due to the small support of the equilibria --- see Figure \ref{fig:random}. Similarly, the degrees of the payoff polynomials have only a small effect on the runtime. 

\begin{figure}[!htpb]
	\centering
	\includegraphics[width=\linewidth]{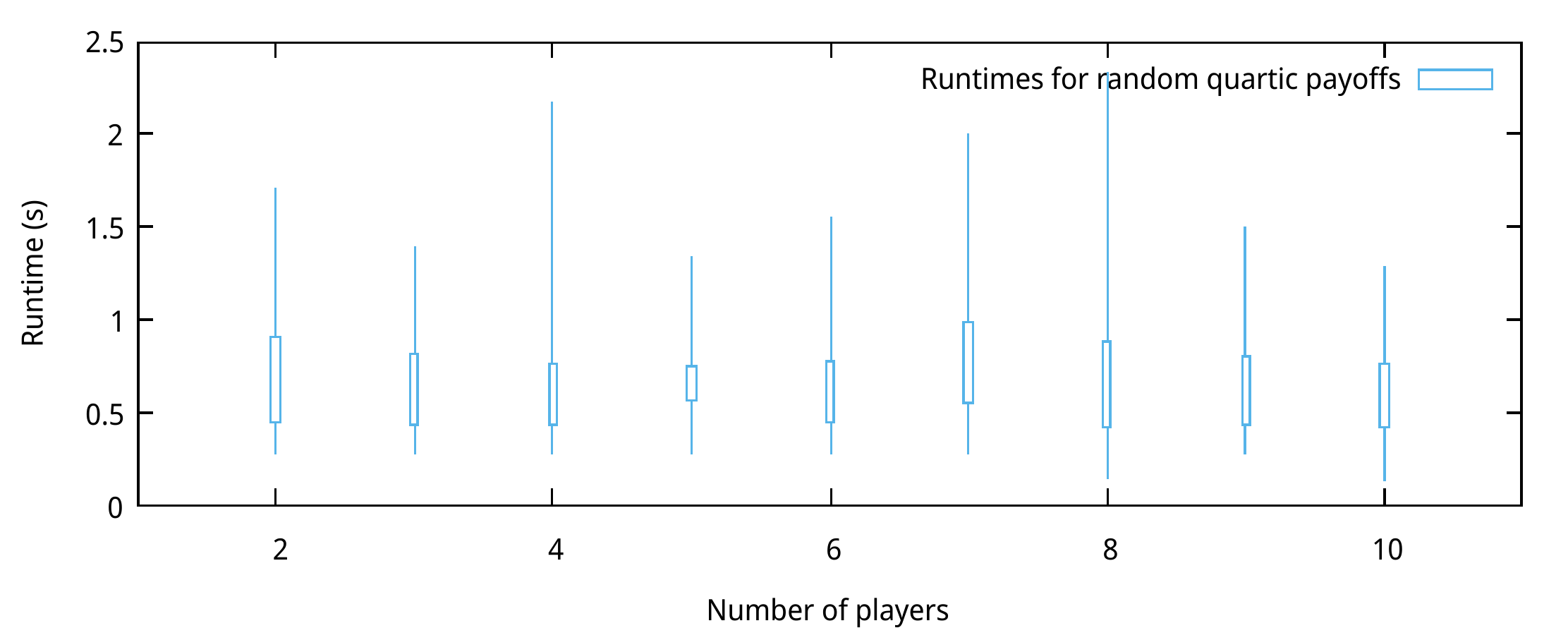}
	\caption{Plot showing the time required to find equilibria of multiplayer polynomial games depending on the number of players. For each player count, the plot shows 100 samples of quartic polynomials with normally distributed random coefficients.}
  \label{fig:random}
\end{figure}

We also conducted the simulation experiment with $100$ samples of polynomial games where each strategy space is the cube $[0,1]^2$. Figure \ref{fig:semibydegree} and Figure \ref{fig:semibyplayer} show the runtimes needed to reach any $10^{-3}$-equilibrium for multiplayer polynomial games with degrees less than $4$ and up to $5$ players.

\begin{figure}[!htpb]
	\centering
	\includegraphics[width=\linewidth]{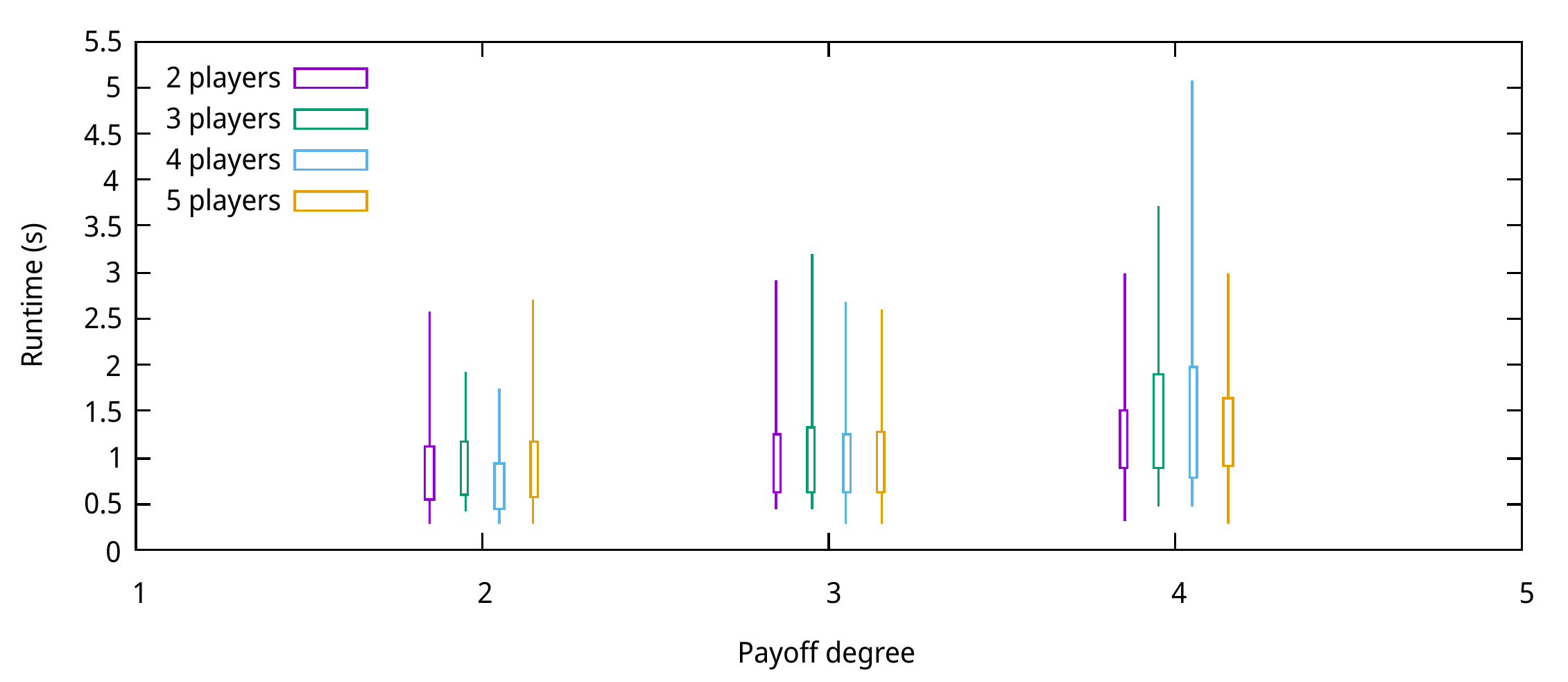}
	\caption{This graph shows the dependence of runtime on the degree of polynomials for games up to $5$ players.}
  \label{fig:semibydegree}
\end{figure}

\begin{figure}[!htpb]
	\centering
	\includegraphics[width=\linewidth]{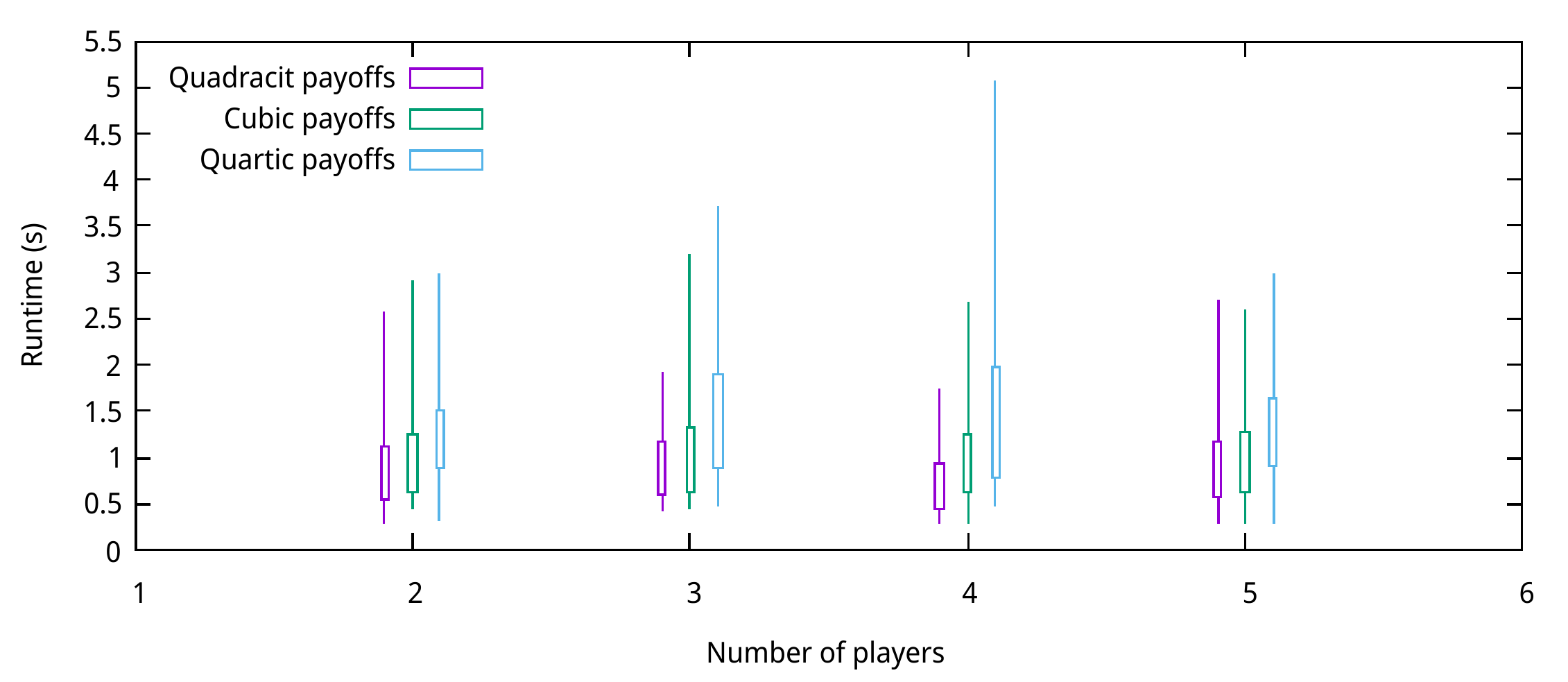}
	\caption{This graph shows the dependence of runtime on the number of players for polynomial games up to degree $4$.}
  \label{fig:semibyplayer}
\end{figure}

\subsection{Using the multiple oracle algorithm to accelerate existing solvers for finite games}
Adding the multiple oracle algorithm on top of solvers such as those implemented in the Gambit library \cite{mckelvey2006gambit} can improve the solution time of large finite general-sum games. Unfortunately, the solvers implemented in Gambit occasionally fail to produce an output or will loop indefinitely. Nevertheless, our preliminary results suggest that the multiple oracle algorithm has the potential to accelerate existing solvers.

We used the global Newton method \cite{govindan2003global} to find equilibria using \texttt{pygambit} as an interface to Gambit, and when the method failed, we used the iterated polymatrix approximation \cite{GOVINDAN20041229} instead. Due to the significant overhead of this approach, our recorded runtimes are much higher than what is theoretically achievable with the multiple oracle algorithm.

\begin{figure}[!htpb]
	\centering
	\includegraphics[width=\linewidth]{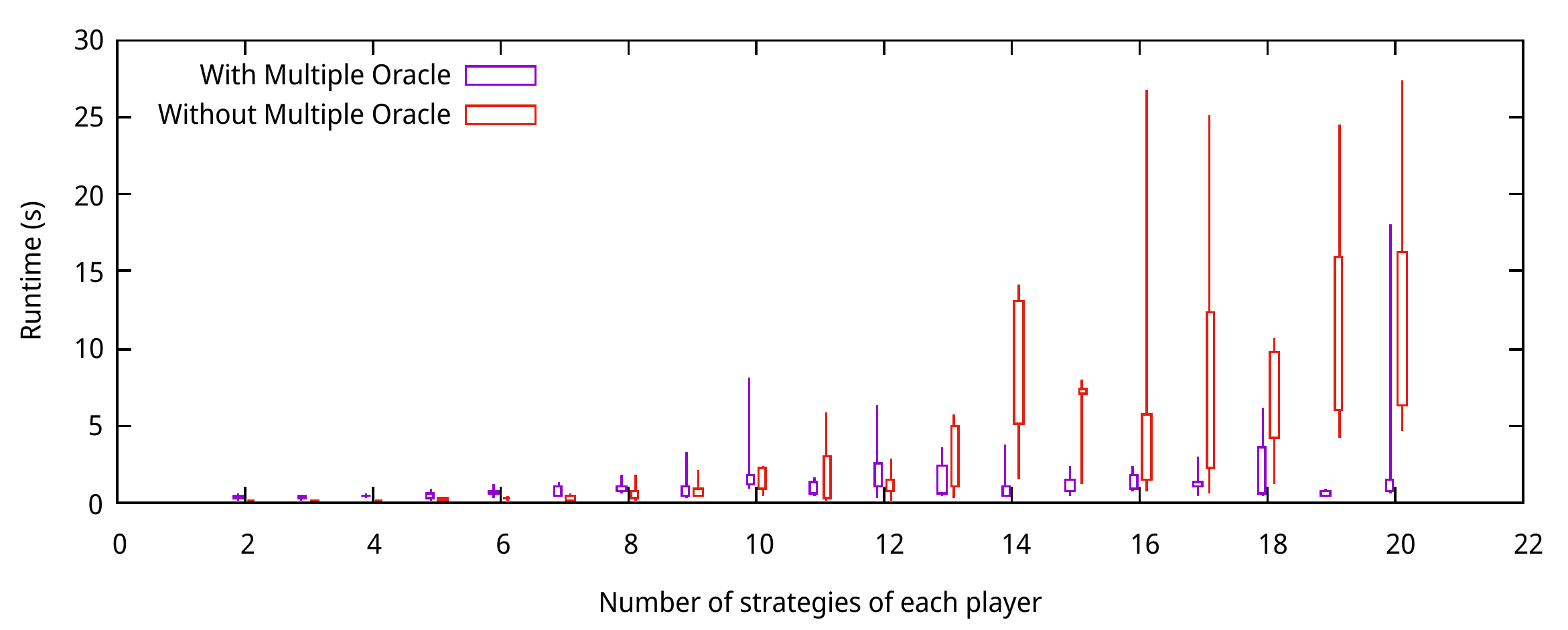}
	\caption{Plot comparing the time required to find equilibria of three-player general-sum finite games with randomly generated payoffs when using algorithms of Gambit with/without the multiple oracle algorithm.}
\end{figure}

\section{Conclusions and Future Research}
The idea of multiple oracle algorithm is to construct a sequence of finite subgames whose equilibria
 approximate the equilibrium of a given continuous game in the Wasserstein metric. We have demonstrated  performance of the algorithm on selected games appearing in current research papers and randomly generated games.  Note that the multiple oracle algorithm makes it possible to approximate the equilibrium of \emph{any} continuous games in the sense of Theorem \ref{thm:convergence}, with the caveat that an individual sequence of equilibria may fail to converge. Although possible in theory (see \cite{Adam21}), this behavior has never been observed in the sample games. Another feature of the algorithm is that the choice of best response oracle and the method for solving finite subgames should be fine-tuned for every particular class of games. A good case in point is the class of polynomial games, which allows for globally optimal solvers for the sub-problem (the best response computation).

 Several examples from the literature show that our method converges fast when the dimensions of strategy spaces are small and the generated subgames are not large. We plan to enlarge the scope of our experiments to include games with many-dimensional strategy spaces. Moreover, the results from Section \ref{sec:exppolyquartic} indicate that more appropriate methods for generating challenging polynomial games should be used to assess the scalability of the multiple oracle algorithm for solving multiplayer polynomial games. 
While we refrained from the detailed discussion of numerous metrics on the space of mixed strategies, we do point out the concept of limit games and equilibria introduced by Fudenberg and Levin \cite{Fudenberg86}. In this connection we plan to study the relation of the underlying convergence to the Wasserstein distance used in this paper.
\bibliographystyle{splncs04}
\bibliography{ref}

\begin{thebibliography}{10}
\providecommand{\url}[1]{\texttt{#1}}
\providecommand{\urlprefix}{URL }
\providecommand{\doi}[1]{https://doi.org/#1}

\bibitem{Adam21}
Adam, L., Hor\v{c}\'ik, R., Kasl, T., Kroupa, T.: Double oracle algorithm for
  computing equilibria in continuous games. In: Proceedings of the AAAI
  Conference on Artificial Intelligence. pp. 5070--5077 (2021)

\bibitem{Bailey2018-to}
Bailey, J.P., Piliouras, G.: Multiplicative weights update in {Zero-Sum} games.
  In: Proceedings of the 2018 {ACM} Conference on Economics and Computation.
  pp. 321--338. EC '18, Association for Computing Machinery, New York, NY, USA
  (Jun 2018)

\bibitem{BasarOlsder99}
Ba{\c s}ar, T., Olsder, G.: Dynamic Noncooperative Game Theory, 2nd Edition.
  Society for Industrial and Applied Mathematics (1999)

\bibitem{Behnezhad2019-fc}
Behnezhad, S., Blum, A., Derakhshan, M., Hajiaghayi, M., Papadimitriou, C.H.,
  Seddighin, S.: Optimal strategies of {B}lotto games: Beyond convexity. In:
  Proceedings of the 2019 {ACM} Conference on Economics and Computation. pp.
  597--616. EC '19, ACM, New York, NY, USA (Jun 2019)

\bibitem{Bertsekas16}
Bertsekas, D.: Nonlinear Programming. Athena Scientific (2016)

\bibitem{Bosansky14}
Bo\v{s}ansk\'y, B., Kiekintveld, C., Lis\'y, V., P\v{e}chou\v{c}ek, M.: An
  exact double-oracle algorithm for zero-sum extensive-form games with
  imperfect information. Journal of Artificial Intelligence Research
  \textbf{51},  829--866 (2014)

\bibitem{Cai16}
Cai, Y., Candogan, O., Daskalakis, C., Papadimitriou, C.: Zero-sum polymatrix
  games: A generalization of minmax. Mathematics of Operations Research
  \textbf{41}(2),  648--655 (2016)

\bibitem{Chasnov20}
Chasnov, B., Ratliff, L., Mazumdar, E., Burden, S.: Convergence analysis of
  gradient-based learning in continuous games. In: Uncertainty in Artificial
  Intelligence. pp. 935--944. PMLR (2020)

\bibitem{Daskalakis09}
Daskalakis, C., Goldberg, P.W., Papadimitriou, C.H.: The complexity of
  computing a {N}ash equilibrium. SIAM Journal on Computing  \textbf{39}(1),
  195--259 (2009)

\bibitem{DunningHuchetteLubin2017}
Dunning, I., Huchette, J., Lubin, M.: {JuMP}: A modeling language for
  mathematical optimization. SIAM Review  \textbf{59}(2),  295--320 (2017).
  \doi{10.1137/15M1020575}

\bibitem{Ferris00}
Ferris, M.C., Munson, T.S.: Complementarity problems in {GAMS} and the {PATH}
  solver. Journal of Economic Dynamics and Control  \textbf{24}(2),  165--188
  (2000)

\bibitem{Fudenberg86}
Fudenberg, D., Levine, D.: Limit games and limit equilibria. Journal of
  Economic Theory  \textbf{38}(2),  261--279 (1986)

\bibitem{Ganzfried21}
Ganzfried, S.: Algorithm for computing approximate {N}ash equilibrium in
  continuous games with application to continuous {B}lotto. Games
  \textbf{12}(2), ~47 (2021)

\bibitem{Gibbs02}
Gibbs, A.L., Su, F.E.: On choosing and bounding probability metrics.
  International statistical review  \textbf{70}(3),  419--435 (2002)

\bibitem{Glicksberg52}
Glicksberg, I.L.: A further generalization of the {K}akutani fixed point
  theorem, with application to {N}ash equilibrium points. Proceedings of the
  American Mathematical Society  \textbf{3},  170--174 (1952)

\bibitem{Golman09}
Golman, R., Page, S.E.: General {B}lotto: Games of allocative strategic
  mismatch. Public Choice  \textbf{138}(3-4),  279--299 (2009)

\bibitem{govindan2003global}
Govindan, S., Wilson, R.: A global {N}ewton method to compute {N}ash
  equilibria. Journal of Economic Theory  \textbf{110}(1),  65--86 (2003)

\bibitem{GOVINDAN20041229}
Govindan, S., Wilson, R.: Computing {N}ash equilibria by iterated polymatrix
  approximation. Journal of Economic Dynamics and Control  \textbf{28}(7),
  1229--1241 (2004)

\bibitem{Hinze08}
Hinze, M., Pinnau, R., Ulbrich, M., Ulbrich, S.: Optimization with {PDE}
  constraints, vol.~23. Springer Science \& Business Media (2008)

\bibitem{Hofbauer06}
Hofbauer, J., Sorin, S.: Best response dynamics for continuous zero-sum games.
  Discrete and Continuous Dynamical Systems--Series B  \textbf{6}(1), ~215
  (2006)

\bibitem{Kamra18}
Kamra, N., Gupta, U., Fang, F., Liu, Y., Tambe, M.: Policy learning for
  continuous space security games using neural networks. In: Thirty-Second AAAI
  Conference on Artificial Intelligence. pp. 1103--1112 (2018)

\bibitem{KroVaVo}
Kroupa, T., Vannucci, S., Votroubek, T.: Separable network games with compact
  strategy sets. In: Bo{\v{s}}ansk{\'y}, B., Gonzalez, C., Rass, S., Sinha, A.
  (eds.) Decision and Game Theory for Security. pp. 37--56. Springer
  International Publishing, Cham (2021)

\bibitem{LarakiLasserre12}
Laraki, R., Lasserre, J.B.: Semidefinite programming for min--max problems and
  games. Mathematical programming  \textbf{131}(1-2),  305--332 (2012)

\bibitem{Lasserre2015}
Lasserre, J.B.: An Introduction To Polynomial And Semi-Algebraic Optimization,
  vol.~52. Cambridge University Press (2015)

\bibitem{Li21}
Li, Z., Wellman, M.P.: Evolution strategies for approximate solution of
  {B}ayesian games. In: Proceedings of the AAAI Conference on Artificial
  Intelligence. vol.~35, pp. 5531--5540 (2021)

\bibitem{mckelvey2006gambit}
McKelvey, R.D., McLennan, A.M., Turocy, T.L.: Gambit: Software tools for game
  theory. Version 16.0.1.  (2016)

\bibitem{McMahan03}
McMahan, H.B., Gordon, G.J., Blum, A.: Planning in the presence of cost
  functions controlled by an adversary. In: Proceedings of the 20th
  International Conference on Machine Learning (ICML-03). pp. 536--543 (2003)

\bibitem{Mertikopoulos2018-ef}
Mertikopoulos, P., Lecouat, B., Zenati, H., Foo, C.S., Chandrasekhar, V.,
  Piliouras, G.: Optimistic mirror descent in saddle-point problems: Going the
  extra (gradient) mile  (Jul 2018)

\bibitem{Mertikopoulos2019}
Mertikopoulos, P., Zhou, Z.: Learning in games with continuous action sets and
  unknown payoff functions. Mathematical Programming  \textbf{173}(1),
  465--507 (2019)

\bibitem{Niu2021}
Niu, L., Sahabandu, D., Clark, A., Poovendran, R.: A game-theoretic framework
  for controlled islanding in the presence of adversaries. In: International
  Conference on Decision and Game Theory for Security. pp. 231--250. Springer
  (2021)

\bibitem{Panaretos20}
Panaretos, V.M., Zemel, Y.: An Invitation To Statistics In {W}asserstein Space.
  Springer Nature (2020)

\bibitem{ParriloIEEE06}
Parrilo, P.: Polynomial games and sum of squares optimization. In: Decision and
  Control, 2006 45th IEEE Conference on. pp. 2855--2860 (2006)

\bibitem{Peyre19}
Peyr{\'e}, G., Cuturi, M.: Computational optimal transport: With applications
  to data science. Foundations and Trends{\textregistered} in Machine Learning
  \textbf{11}(5-6),  355--607 (2019)

\bibitem{Rehbeck18}
Rehbeck, J.: Note on unique {N}ash equilibrium in continuous games. Games and
  Economic Behavior  \textbf{110},  216--225 (2018)

\bibitem{roussillon2021scalable}
Roussillon, B., Loiseau, P.: Scalable optimal classifiers for adversarial
  settings under uncertainty. In: International Conference on Decision and Game
  Theory for Security. pp. 80--97. Springer (2021)

\bibitem{SteinOzdaglarParrilo08}
Stein, N.D., Ozdaglar, A., Parrilo, P.A.: Separable and low-rank continuous
  games. International Journal of Game Theory  \textbf{37}(4),  475--504 (2008)

\bibitem{ipopt}
W{\"a}chter, A., Biegler, L.T.: On the implementation of an interior-point
  filter line-search algorithm for large-scale nonlinear programming.
  Mathematical Programming  \textbf{106}(1),  25--57 (Mar 2006).
  \doi{10.1007/s10107-004-0559-y},
  \url{https://doi.org/10.1007/s10107-004-0559-y}

\bibitem{weisser2019polynomial}
Weisser, T., Legat, B., Coey, C., Kapelevich, L., Vielma, J.P.: Polynomial and
  moment optimization in {Julia} and {JuMP}. In: JuliaCon (2019),
  \url{https://pretalx.com/juliacon2019/talk/QZBKAU/}

\bibitem{Xu21}
Xu, L., Perrault, A., Fang, F., Chen, H., Tambe, M.: Robust reinforcement
  learning under minimax regret for green security. In: Uncertainty in
  Artificial Intelligence. pp. 257--267. PMLR (2021)

\bibitem{Loiseau19}
Yasodharan, S., Loiseau, P.: Nonzero-sum adversarial hypothesis testing games.
  In: Advances in Neural Information Processing Systems. pp. 7310--7320 (2019)

\end{thebibliography}

\end{document}